\documentclass[10pt,onecolumn,onesided]{IEEEtran}
\usepackage[utf8]{inputenc}
\usepackage{lipsum}
\usepackage{eulervm}
\usepackage{amsmath}
\usepackage{amsthm}
\usepackage{import}
\usepackage{geometry}
\usepackage{mathrsfs} 
\usepackage{amsfonts}
\usepackage{amssymb}
\usepackage{palatino,color}
\usepackage{fullpage}
\usepackage{times}
\usepackage{array}
\usepackage{tkz-graph}
\usepackage{kantlipsum}
\usepackage{fancyhdr,graphicx,amsmath,amssymb}
\usepackage[ruled,vlined]{algorithm2e}
\include{pythonlisting}
\usepackage{wrapfig}
\graphicspath{{./},{./figures/}}
\usepackage{tikz}
 
\usepackage[utf8]{inputenc}
\usepackage{tikz}
\usetikzlibrary{arrows}
\usepackage{cite}
\usepackage{tikz,bm,color}
\usetikzlibrary{automata}
\usepackage{comment}
\usetikzlibrary{shapes.geometric}
\usetikzlibrary{positioning}
\usetikzlibrary{circuits.logic.US,circuits.logic.IEC,fit}

\usepackage{xparse}
\NewDocumentCommand\DownArrow{O{2.0ex} O{black}}{%
   \mathrel{\tikz[baseline] \draw [<-, line width=0.5pt, #2] (0,0) -- ++(0,#1);}
}
\tikzset{every picture/.style={draw=black,thick}}
\usepackage{caption}
\usepackage{hyperref}
\usepackage{float}
\usepackage[font=small,skip=-10pt]{caption}
\usepackage{comment,accents}
\usepackage[export]{adjustbox}
\usepackage{enumerate,graphicx,enumitem,enumerate}
\usepackage{amsmath}

\usepackage{amsmath,amssymb,euscript,yfonts,psfrag,latexsym,dsfont,graphicx,bbm,color,amstext,wasysym,subfig,flushend,parskip,soul,bm}
\usepackage{amsthm}
\graphicspath{{./},{./figures/}}

\newcommand{\ignore}[1]{}

\newcommand{\sbprgraph}{}

\definecolor{grey}{rgb}{0.6,0.3,0.3}
\definecolor{lgrey}{rgb}{0.9,.7,0.7}
\definecolor{lightblue}{rgb}{0.5,.5,.8}

\usepackage{graphicx}
\def\spacingset#1{\def\baselinestretch{#1}\small\normalsize}
\spacingset{1.8}

\newtheorem{thm}{Theorem}
\newtheorem{cor}[thm]{Corollary}

\newtheorem{lemma}[thm]{Lemma}
\newtheorem{prop}[thm]{Proposition}

\title{\huge Macroscopic network circulation for planar graphs}
\author{Fariba Ariaei\textsuperscript{*}, \and Zahra Askarzadeh\textsuperscript{*}, \and Yongxin Chen, \and Tryphon T. Georgiou\thanks{F.\ Ariaei, Z.\ Askarzadeh, and T.\ T. Georgiou are with the Department of Mechanical and Aerospace Engineering, University of California, Irvine, CA; rfu2@uci.edu, ataghvae@uci.edu, tryphon@uci.edu}
\thanks{Y.~Chen is with the School of Aerospace Engineering, Georgia Institute of Technology, Atlanta, GA 30332; {yongchen@gatech.edu}}\thanks{\noindent\textsuperscript{*}F.\ Ariaei and Z.\ Askarzadeh contributed equally as first authors.}
\thanks{Research supported in part by the NSF under grants 1807664, 1839441, 1901599, and the AFOSR under FA9550-20-1-0029.}}

\makeatletter
\tikzset{my loop/.style =  {to path={
  \pgfextra{}
  [looseness=4,min distance=3mm,->]
  \tikz@to@curve@path},font=\sffamily\small
  }}  
\makeatletter 

\begin{document}
\maketitle

\begin{abstract} The analysis of networks, aimed at suitably defined functionality, often focuses on partitions into subnetworks that capture desired features. Chief among the relevant concepts is a {\em 2-partition}, that underlies the classical Cheeger inequality, and highlights a constriction (bottleneck) that limits accessibility between the respective parts of the network. In a similar spirit, the purpose of the present work is to introduce a new concept of {\em maximal global circulation} and to explore  {\em 3-partitions} that expose this type of macroscopic feature of networks.
Herein, graph circulation is motivated by transportation networks and probabilistic flows (Markov chains) on graphs. Our goal is to quantify the large-scale imbalance of network flows and delineate key parts that mediate such global features. While we introduce and propose these notions in a general setting, in this paper, we only work out the case of planar graphs. We explain that a scalar potential can be identified to encapsulate the concept of circulation, quite similarly as in the case of the curl of planar vector fields. Beyond planar graphs, in the general case, the problem to determine global circulation remains at present a combinatorial problem.\\[-.35in]
\end{abstract}

\section{Introduction}
Time asymmetry of traffic flow in city streets is unmistakeable. Specifically, traffic flows in one direction around city squares and, often, in one-way in many city streets as well. Yet, from a macroscopic vantage point, circulation may or may not be evident. Flux from one part of town to another may average out with flux in the opposite direction. When this is not the case, it is of interest to identify the nature and to quantify any large scale imbalance in global circulation. What we seek in the present article is precisely such a notion of {\em macroscopic circulation} that, depending on the network and flow conditions, captures the flow asymmetry that points to a preference in directionality while traversing the graph.

Perhaps, circulation is nowhere more apparent than in air currents at the planetary scale.  The vorticity, locally as well as at earth-scale, very much as in planar vector fields, can be quantified by a suitably defined scalar potential and, as we will explain, this scalar potential helps quantify maximal circulation macroscopically. A purpose of this work is to define a notion of macroscopic circulation on discrete spaces, namely graphs, and utilize a similar construct of a scalar potential for its computation.

At present, in full generality, macroscopic circulation remains a combinatorial problem. In particular, it is not known whether it relates in any way to spectral properties of graph Laplacians (as is the case for the Cheeger constant, that quantifies graph bottlenecks) or to other topological characteristics of the graph. Thus, in this work, we proceed to explore the special case of {\em embedded planar graphs} and, taking advantage of insights from the well-known Helmholtz-Hodge decomposition of vector fields, we explain how macroscopic circulation can be effectively computed by determining the range of values of a scalar potential with support on the nodes of the dual graph.

The mathematical setting that we contemplate model flows on graphs is that of a stationary discrete-time Markov model, where the probability flux represents a vector field on the network of the nodes and edges of the Markov chain. 
In this setting, the famous Cheeger inequality relates the likelihood of transitioning between two parts in a 2-partition of the nodes, across, in either direction, as well as the rate of mixing, to spectral properties  of the graph Laplacian. As a 2-partition aims to capture bottlenecks that impede mixing, the Cheeger constant duly takes into account the size of boundary between the two parts. 

In a similar manner, contemplating the elements of circulatory imbalance, we are led naturally to a 3-partition of the network. 
After all, in a 2-partition, the probability current across the boundary, at stationarity, balances out.
Tell-tale signs of circulatory asymmetry requires at least three parts.
In general, flux-imbalance at the micro or macro level may manifest itself when more that two components interact and exchange ``mass.''
For a 3-state partitioning of a network into parts $A$, $B$, and $C$, the net flow from $A\to B$ (which is considered positive when the net flux is in the direction of $B$), by detailed-balance, must equal to the net flux from $B\to C$, and must also equal the net flux from $C\to A$. Thereby, the asymetry manifests itself as a network circulation current. For reasons similar to those underlying the Cheeger constant, careful consideration of the size and regularity of the boundary between the three parts is warranted and need
to be duly restricted.

The structure of the paper is as follows. In section \ref{sec:markov}, we discuss connection between probability currents and flow fields on graphs. In Section \ref{sec:circulation}, we highlight the concept of macroscopic circulation in the context of Markov chains. The setting of Markov chains is not restricted by the dimensionality of possible embedding of the respective graph, but the formulation of global circulation in general requires further refinement. In Section \ref{sec:planarity}, we discuss planar graph and the decomposition of flows accordingly. In section \ref{sec:partitioning}, we introduce an algorithm for calculating scalar potential supported on the dual graph and propose a method for partitioning the graph into three parts and calculating macroscopic circulation. The issue of embedding is revisited in Section \ref{sec:embedding} where it is explained that a given graph may have non-equivalent embeddings, leading to different values for the macroscopic circulation. In Section \ref{sec:examples}, we detail additional illuminating examples.

\section{Probability currents \& flow fields on graphs}
\label{sec:markov}

Herein we explain the correspondence between flows on graphs and probability currents induced by a Markov structure. Flow fields represent the analogue of vector fields, though we avoid the term ``vector'' to highlight the difference with ordinary vector fields on manifolds.

Consider a time-homogeneous, discrete-time, $N$-state finite Markov chain $X_t$, with $t\in\mathbb N$, with states $\mathcal V=\{v_1,\ldots,v_N\}$, comprised of the nodes of a network, and transition probabilities $\pi_{v_i,v_j}$, i.e., 
\[
\mathbb P\left\{X_{t+1}=v_j\mid X_{t}= v_i\right\}=\pi_{v_i,v_j}.
\]
We assume that the Markov chain is ergodic and hence, irreducible and aperiodic.
Thus, the matrix $\Pi:=[\pi_{v_i,v_j}]_{v_i,v_j=1}^N$  has non-negative entries and is such that $\Pi\mathds{1}=\mathds{1}$, where $\mathds{1}$ denotes a column vector with all entries equal to $1$. The ergodicity assumption implies that for a sufficiently large integer $k$ (e.g., $k=N$), $\Pi^k$ has all entries positive. The dimensionality of vectors and matrices will be explicit, unless their dimension is clear from the context.

The Markov chain is associated to a simple graph $\mathcal G:=(\mathcal V,\mathcal E)$,
where the (directed) edge set $\mathcal E$ is specified by the allowed transitions, i.e., 
\[
\mathcal E=\{e=(v_i,v_j)\,|\, \pi_{v_i,v_j}\neq 0\}.
\]
Throughout we consider $\mathcal G$ to have only one edge for any ordered pair of nodes, i.e., that it is simple.  Further, we consider $\mathcal G$ to be strongly connected that follows from the ergodicity assumption of the Markov chain.

Let now $\pi=[\pi_{v_i}]_{v_i=1}^N$ denote the (column) stationary probability vector of the Markov chain. Thus, 
\[
\pi^T\Pi=\pi^T
\] 
and $\pi^T$ is the (unique left/row Frobenious-Perron) eigenvector of $\Pi$ with eigenvalue $1$. Throughout, $^T$ denotes transposition.
The entries of
\begin{equation}\label{eq:P}
P:= {\rm diag}(\pi_{v_1},\dots, \pi_{v_N})\Pi
\end{equation}
represent probability current $p_{v_i,v_j}=\pi_{v_i}\pi_{v_i,v_j}$ from vertex $v_i$ to $v_j$. Probability currents quantify {\em flux} on $\mathcal G$.

Our aim in this paper is to identify (large scale) imbalance in the {\em net flux} across $\mathcal G$, and to this end, we will be working mostly with the anti-symmetric part of $P$ (modulo a factor of $1/2$)
\begin{equation}
F = P-P^T.
\label{netflux}
\end{equation}
This retains information on only local flux imbalance between any two nodes.
Note that since, $\Pi\mathds{1}=\mathds{1}$, it follows that $P\mathds{1}=P^T\mathds{1}$, and therefore, that
$F\mathds{1}=\mathbf 0$ (the zero vector) as well. 

We view the matrix $F$ as representing a ``divergence free'' (i.e., with no sources) {\em flow (``vector'') field} on $\mathcal G$. Besides the fact that
\begin{subequations}
\begin{align}\label{eq:F1}
&F=-F^T\\\label{eq:F2}
&F\mathds{1}=\mathbf 0,
\end{align}
the positive part of $F$, namely,
\[F_+:=\left[ \max\{F_{ij},0\}\right]_{i,j=1}^n,
\]
has entries that are $\leq$ to those of $P$, and since $\mathds{1}^TP\mathds{1}=1$,
\begin{align}
&\mathds{1}^TF_+\mathds{1}\leq 1.\label{eq:F3}
\end{align}
\end{subequations}

It turns out that (\ref{eq:F1}-\ref{eq:F3}) characterize divergence-free flow fields on graphs. I.e., antisymmetric matrices with the above properties originate from a Markovian probability structure. We state the precise result below.\\[-.05in]

\begin{prop}
Consider an $N\times N$ matrix $F$ and assume that (\ref{eq:F1}-\ref{eq:F3}) hold. In case
\eqref{eq:F3} holds with equality, assume that $\mathds{1}^TF_+$ has all entries positive. Then, $F$ originates as a divergence-free flow-field on a graph $\mathcal G=(\mathcal V,\mathcal E)$, with $|\mathcal V|=N$, associated with a Markov probability structure.\\[-.35in]
\end{prop}
\begin{proof}
If \eqref{eq:F3} holds with equality, let $P=F_+$, otherwise
define
\begin{equation}
P :=  M + F_+, 
\label{P}
\end{equation}
for a symmetric matrix $M=M^T$, of the same size, with nonnegative entries such that
\[ \mathds{1}^{T} M \mathds{1} = 1-\mathds{1}^{T} F_{+} \mathds{1},
\]
ensuring that $\pi^T:=\mathds{1}^TP$ has all entries positive.
This is clearly possible from the standing assumptions. Now, verify that
\begin{equation}\label{eq:Pichoice}
\Pi={\rm diag}(\pi_{v_1},\dots, \pi_{v_N})^{-1}P
\end{equation}
is a transition probability matrix that leads to the divergence-free flow field $F$.
Specifically,
 i) $\Pi$ has non-negative entries. ii) In view of $\pi^T= \mathds{1}^TP$ and \eqref{eq:Pichoice}, $\pi^T\Pi=\pi^T$ holds.
iii)
Note that $F=F_+-F_+^T$ and hence, $F_+\mathds{1}=F_+^T\mathds{1}$ from \eqref{eq:F2}. It follows that $P\mathds{1}=P^T\mathds{1}$, and from \eqref{eq:Pichoice} the definition $\pi^T= \mathds{1}^TP$, that
$\Pi\mathds{1}=\mathds{1}$.
iv) Lastly, $P-P^T=F_+-F_+^T=F$.
\end{proof}

\section{Macroscopic circulation on graphs}
\label{sec:circulation}

Consider an $N\times N$ antisymmetric matrix $F$ of net fluxes that defines a divergence-free\footnote{If this is not the case, we replace $F$ by its restriction on the complement of the range of $\mathds{1}\mathds{1}^T$, namely, $(I-\frac{1}{N}\mathds{1}\mathds{1}^T )F(I-\frac{1}{N}\mathds{1}\mathds{1}^T )$, so that $F\mathds{1}=0$.} flow field on a (simple) graph $\mathcal G$. 
We seek a suitable definition of (maximal) {\em macroscopic circulation} by partitioning the states into three subsets $A$, $B$, and $C$, in such a way so as to maximize the flux between the parts.
We discuss first the simplest case, of three states, and proceed to define the concept of circulation and flow-density in general.

\subsubsection*{Three-state example}
We consider a three-state Markov chain ($N=3$) in Fig. \ref{fig:fluxes}, where for convenience we label the three nodes as $A,B,C$, i.e., $\mathcal V=\{A,B,C\}$.
The net-flux matrix on $\mathcal G$ (anti-symmetric part of the probability current matrix $P$, modulo a factor of $1/2$) is
\[
F =\left[\begin{matrix}0&-\gamma&\gamma\\\gamma&0&-\gamma\\-\gamma&\gamma&0\end{matrix}\right],
\]
with the directionality encoded in the sign of $\gamma$.
Obviously, the off-diagonal entries of F must have the same magnitude, since $F\mathds{1}=0$. Evidently, the value $|\gamma|$ quantifies circulation in this example.
The weighted oriented graph of net fluxes is shown in Fig. \ref{fig:netfluxes}.
For the case of a three-state Markov chain, a close-form expression for $\gamma$ can be obtained in terms of $\Pi$, though this is immaterial and not to be expected in general. \hfill $\Box$
\begin{figure}[ht]
\begin{center}
\begin{tikzpicture}[->,>=stealth',shorten >=1pt,thick,scale=.7]
\tikzset{edge/.style = {->,> = latex'}}
\node[state] (C) at  (0,0) {$C$};
\node[state] (A) at  (-2.3,-3.5) {$A$};
\node[state] (B) at  (2.3,-3.5) {$B$};

\Loop[dist=1.2cm,dir=NO,label=$p_{CC}$,labelstyle=above](C) 
\Loop[dist=1.2cm,dir=SOWE,label=$p_{AA}$,labelstyle=below left](A) 
\Loop[dist=1.2cm,dir=SOEA,label=$p_{BB}$,labelstyle=below right](B)  
\draw
(C) edge[bend right=10,auto=right,->] node {$p_{CB}\hspace*{-5pt}$} (B)  
(B) edge[bend right=10,auto=right,->] node {$p_{BC}$} (C)
(A) edge[bend right=10,auto=right,->] node {$p_{AB}$} (B)  
(B) edge[bend right=10,auto=right,->] node {$p_{BA}$} (A)
(C) edge[bend right=10,auto=right,->] node {$p_{CA}$} (A)  
(A) edge[bend right=10,auto=right,->] node {\hspace*{-5pt}$p_{AC}$} (C);  

\end{tikzpicture}
\vspace*{.2in}
\caption{Probability currents in a 3-state Markov chain.}\label{fig:fluxes}
\end{center}
\end{figure}

\begin{figure}[ht]
\begin{center}
\begin{tikzpicture}[scale=.7]
\tikzset{edge/.style = {->,> = latex'}}
\node[state] (C) at  (0,0) {C};
\node[state] (A) at  (-2.3,-3.5) {A};
\node[state] (B) at  (2.3,-3.5) {B};

\draw
(A) edge[bend left=0,auto=left,->] node {$p_{AC}-p_{CA}=\gamma$} (C)
(C) edge[bend left=0,auto=left,->] node {$\gamma=p_{CB}-p_{BC}$} (B)
(B) edge[bend left=0,auto=left,->] node {$\gamma=p_{BA}-p_{AB}$} (A);  

\end{tikzpicture} 
\vspace*{.2in}
\caption{Weighted oriented graph of net fluxes.}\label{fig:netfluxes}
\end{center}
\end{figure}

%

%

\subsubsection*{General case}
We now consider the general case with $N$ states, as before, and data for net-flux between nodes in $F$.
We seek partitioning the graph into three subsets of nodes 
\[
\mathcal A,\mathcal B, \mathcal C\subset \mathcal V
\]
that are pairwise non-intersecting with
\[
\mathcal V=\mathcal A\cup \mathcal B\cup \mathcal C.
\]
Such a triple of subsets of $\mathcal V$ will be referred to as a 3-partition.

Define the characteristic (column) vector $I_{\mathcal S}$ of a set $\mathcal S\subseteq \mathcal V$, with $\mathcal V$ ordered, as follows: the $v^{th}$ entry of $I_{\mathcal S}$ is equal to $1$ when $v\in\mathcal S$ and $0$ otherwise. It is convenient to define entry-wise Boolean addition and multiplication of characteristic vectors, $\oplus$ and $\circ$, respectively, and also the notation $\mathds{1},\mathbf{0}$ to denote the (column) vectors with all $1$'s and all $0$'s, respectively.
\begin{lemma}
 $(\mathcal A,\mathcal B, \mathcal C)$ is a 3-partition of $\mathcal V$ if and only if
\begin{subequations}
\begin{align}\label{eq:sum}
&I_\mathcal A\oplus I_\mathcal B\oplus I_\mathcal C=\mathds{1},\\\label{eq:pairwiseproduct}
&I_\mathcal A\circ I_\mathcal B=I_\mathcal B\circ I_\mathcal C=I_\mathcal C\circ I_\mathcal A=\mathbf{0},
\end{align}
\end{subequations}
\end{lemma}
\begin{proof} Relation \eqref{eq:sum} is equivalent to $\mathcal V=\mathcal A\cup \mathcal B\cup \mathcal C$. Then, \eqref{eq:pairwiseproduct} is equivalent to the pair-wise non-intersection condition.
\end{proof}
Given a flow field (net-probability flux) matrix $F$, as before, and a 3-partition $(\mathcal A,\mathcal B, \mathcal C)$ of the (ordered) vertex set $\mathcal V$, then
$I_{\mathcal A}^TFI_{\mathcal B}$
is the (signed) flux directed from $\mathcal A$ to $\mathcal B$. That is, if $I_{\mathcal A}^TFI_{\mathcal B}<0$, the net-flux, summed over all edges connecting directly $\mathcal A$ and $\mathcal B$, is directed from $\mathcal B$ into $\mathcal A$. Thus,
\[I_{\mathcal A}^TFI_{\mathcal B}=-I_{\mathcal B}^TFI_{\mathcal A},
\]
while the absolute value $|I_{\mathcal A}^TFI_{\mathcal B}|$ is the total net-flux between the two parts.

\begin{lemma}
For any 3-partition $(\mathcal A,\mathcal B, \mathcal C)$ of the (ordered) vertex set $\mathcal V$,
\[
I_{\mathcal A}^TFI_{\mathcal B}=I_{\mathcal B}^TFI_{\mathcal C}=I_{\mathcal C}^TFI_{\mathcal A}.
\]
\end{lemma}
\begin{proof}
Note that $I_{\mathcal A}^TFI_{\mathcal A}=0$, since $F$ is antisymmetric, and that $I_{\mathcal A}\oplus I_{\mathcal B\cup \mathcal C}=I_{\mathcal A\cup\mathcal B\cup \mathcal C}=\mathds{1}$. Then,
\begin{align*}
I^T_{\mathcal A}FI_{\mathcal B} + I^T_{\mathcal A}FI_{\mathcal C} &=I^T_{\mathcal A}FI_{\mathcal B\cup \mathcal C}\\
&=I^T_{\mathcal A}FI_{\mathcal B\cup \mathcal C} +I_{\mathcal A}^TFI_{\mathcal A}\\
&=I^T_{\mathcal A}F\mathds{1}=0.
\end{align*}
Thus,
\begin{align*}
I^T_{\mathcal A}FI_{\mathcal B} &=- I^T_{\mathcal A}FI_{\mathcal C}\\
&=\phantom{-}I^T_{\mathcal C}FI_{\mathcal A}.
\end{align*}
And, similarly, $I^T_{\mathcal B}FI_{\mathcal A} = I^T_{\mathcal C}FI_{\mathcal B}$.
\end{proof}

Now, one is naturaly led to define the circulation
\[
c(\mathcal A,\mathcal B, \mathcal C):=|I_{\mathcal A}^TFI_{\mathcal B}|
\]
associated to any given 3-partition and, accordingly, the maximal macroscopic circulation
\[
c_{\rm max}:=\max_{\mbox{3-partitions}} c(\mathcal A,\mathcal B, \mathcal C).
\]
Evidently, $c(\mathcal A,\mathcal B, \mathcal C)$ depends on the partition as well as the ``divergence-free'' flow field on the graph $\mathcal{G}=(\mathcal V,\mathcal E)$ that is specified by the skew symmetric matrix $F$. Herein, we prefer to let $F$ be specified from the context, instead of using a more cumbersome notation such as $c_F(\mathcal A,\mathcal B, \mathcal C)$.

A moment's reflection reveals that these concepts do not take into account the topology of the partition. More specifically, the nature and size of the boundary between the parts of the partition may be relevant to the type of global feature we may want to capture. We highlight this point with the following example, and then return to define normalized notions of macroscopic circulation.

\begin{figure}
\centering
\begin{tikzpicture}
\tikzset{vertex/.style = {shape=circle,draw,minimum size=0.5em}}
\tikzset{edge/.style = {->,> = latex'}}
\node[vertex, fill=green!60] (1) at  (-0.75,1.299) {1};
\node[vertex, fill=red!60] (2) at  (0.75,1.299) {2};
\node[vertex, fill=blue!60] (3) at  (1.5,0) {3};
\node[vertex, fill=green!60] (4) at  (0.75,-1.299) {4};
\node[vertex, fill=red!60] (5) at (-0.75,-1.299) {5};
\node[vertex, fill=blue!60] (6) at (-1.5,0) {6};

\draw[->] (1) to (2);
\draw[->] (2) to (3);
\draw[->] (3) to (4);
\draw[->] (4) to (5);
\draw[->] (5) to (6);
\draw[->] (6) to (1);
\end{tikzpicture}
\vspace*{.25in}
\caption{3-partition into non-contiguous pairs.}
\label{fig:HexGraph a}
\vspace*{.25in}
\begin{tikzpicture}
\tikzset{vertex/.style = {shape=circle,draw,minimum size=0.5em}}
\tikzset{edge/.style = {->,> = latex'}}
\node[vertex, fill=green!60] (1) at  (-0.75,1.299) {1};
\node[vertex, fill=green!60] (2) at  (0.75,1.299) {2};
\node[vertex, fill=red!60] (3) at  (1.5,0) {3};
\node[vertex, fill=red!60] (4) at  (0.75,-1.299) {4};
\node[vertex, fill=blue!60] (5) at (-0.75,-1.299) {5};
\node[vertex, fill=blue!60] (6) at (-1.5,0) {6};

\draw[->] (1) to (2);
\draw[->] (2) to (3);
\draw[->] (3) to (4);
\draw[->] (4) to (5);
\draw[->] (5) to (6);
\draw[->] (6) to (1);
\end{tikzpicture}  
\vspace*{.25in}
\caption{3-partition into contiguous pairs}
\label{fig:HexGraph b}
\vspace*{.25in}
\begin{tikzpicture}
\tikzset{vertex/.style = {shape=circle,draw,minimum size=0.5em}}
\tikzset{edge/.style = {->,> = latex'}}
\node[vertex, fill=blue!60] (1) at  (-0.75,1.299) {1};
\node[vertex, fill=red!60] (2) at  (0.75,1.299) {2};
\node[vertex, fill=green!60] (3) at  (1.5,0) {3};
\node[vertex, fill=green!60] (4) at  (0.75,-1.299) {4};
\node[vertex, fill=green!60] (5) at (-0.75,-1.299) {5};
\node[vertex, fill=green!60] (6) at (-1.5,0) {6};

\draw[->] (1) to (2);
\draw[->] (2) to (3);
\draw[->] (3) to (4);
\draw[->] (4) to (5);
\draw[->] (5) to (6);
\draw[->] (6) to (1);
\end{tikzpicture}  
\vspace*{.25in}
\caption{ 3-partition into contiguous pairs}
\label{fig:HexGraph c}

\end{figure}
\subsubsection*{Six-state example}
Consider hexagonal planar graph in Figures \ref{fig:HexGraph a} and \ref{fig:HexGraph b}.
The two figures display color-coded 3-partitions, where in the first, pairs that constitute each of the three parts are not contiguous, whereas in the second, pairs of nodes in each of the three parts are neighboring and connected. The net-flux matrix $F$ is 
\begin{equation}
F =\left[\begin{matrix}0&1&0&0&0&-1\\-1&0&1&0&0&0\\0&-1&0&1&0&0\\0&0&-1&0&1&0\\0&0&0&-1&0&1\\1&0&0&0&-1&0\end{matrix}\right]
\label{eq:hexGraphNetFluxMatrix}
\end{equation}
Let ${\mathcal A}$ denote the green set of nodes and ${\mathcal B}$ the red.
In the first of these two 3-partitions, $I^T_{\mathcal A}FI_{\mathcal B}=2$, whereas in the
second $I^T_{\mathcal A}FI_{\mathcal B}=1$. The difference between the two is that
the ``circulatory flux'' in the first 3-partition is counted {\em twice} due to the fact that the boundary between the parts ${\mathcal A}$ and ${\mathcal B}$, denoted by $\partial^{\mathcal A}_{\mathcal B}$ has size $2$, as it consists of two edges, and is traversed ``twice'' by any complete transport around a cycle. In the second choice for a 3-partition, $\partial^{\mathcal A}_{\mathcal B}=1$. 
\hfill $\Box$
%
%

The above examples suggest normalizing the flux between any two parts of a partition, by dividing by the size of the corresponding boundaries. That is, normalizing $I^T_{\mathcal A}FI_{\mathcal B}$ by dividing with the size of the boundary, namely,  $|\partial^{\mathcal A}_{\mathcal B}|$ which denotes the cardinality of the set of edges between the two parts ${\mathcal A}$, ${\mathcal B}$,
brings us to a notion of {\em density flux} associated with the boundary separating two parts of any partition,
\[
f(\partial^{\mathcal A}_{\mathcal B}):=\frac{|I_{\mathcal A}^TFI_{\mathcal B}|}{|\partial^{\mathcal A}_{\mathcal B}|}.
\]
Accordingly, the minimal density flux of the partition, is
\begin{align}
f_{\rm min}(\mathcal A,\mathcal B,\mathcal C)&:=\min\{f(\partial^{\mathcal A}_{\mathcal B}),
f(\partial^{\mathcal B}_{\mathcal C}),
f(\partial^{\mathcal C}_{\mathcal A})\},
\label{secondmethod}
\end{align}
and similarly for the maximal. This approach leads us to a combinatorial problem. In fact, for a graph with $n$ vertices, the total number of possible cases to consider for solving \eqref{secondmethod} is equal to Stirling number of the second kind, $\it{S}(n,3)$ \cite{Stirling2015}.

We are interested in circulation defined on 3-partitions of a graph with partitions having admissible boundaries. We will focus on planar graphs and explain how to compute suitable notions of macroscopic circulation via a scalar potential supported on the dual graph. As an example, a partitioning for the hexagonal planar graph has been shown in Fig. \ref{fig:HexGraph c}, in which all of the partitions are connected. We will show that in general under some
assumptions this method guarantees connectivity of the partitions.

\section{Planar graphs and network circulation}\label{sec:planarity}

We assume that geographic proximity of nodes is dictated by an actual embedding of a graph into a linear metric space, specifically $\mathbb R^2$. Graphs that can be embedded in $\mathbb R^2$, without intersection of edges, are called {\em planar}. In this case flow fields have astrong resemblance to planar vector fields. 

For planar vector fields there is a well known decomposition into gradient flow and curl that captures circulation. In fact, circulation can be conveniently quantified by a scalar potential. In a similar manner, for planar graphs, circulation relates to a scalar potential on the vertex set of a dual graph as we will explain shortly.\footnote{The dual graph $\mathcal G^{*}$ of a planar graph $\mathcal G$ is a planar graph that each of its vertices corresponds to a face of $\mathcal G$ and each of whose faces corresponds to a graph vertex of $\mathcal G$. Two nodes in $\mathcal G^{*}$ are connected by an edge if the corresponding faces in $\mathcal G$ have an edge as a boundary.}
We proceed to review some facts on planar graphs, as well as elements of the Helmholtz-Hodge decomposition of vector fields that provides insight into the corresponding decomposition of flow fields on graphs.\\[.2in]
\vspace{-1cm}
\subsection{Planar graphs}

A graph is referred to as {\em planar} when it can be drawn on the plane in a way that no edges cross each other and intersect only at their endpoints (vertices). The study of planar graphs goes back to Euler who showed that, for simple and connected planar graphs on the simply connected space $\mathbb{R}^2$, i.e., with zero genus $g$, the Euler characteristic\footnote{Euler characteristic is a topological invariant, a number that describes a topological space's shape or structure regardless of the way it is bent.} $\chi(g)=2$ in (\ref{eq:euler})
\begin{equation} \label{eq:euler}
\begin{split}
|\mathcal V|-|\mathcal E|+|\mathcal F| & = \chi(g) \\
\chi(g)& =2-2g
\end{split}
\end{equation}
where $\mathcal F$ is the face{\footnote{The exterior of the graph needs to be counted as a face, and it is referred to as the outside face.} set. Interestingly,  Euler formula is not sufficient to ensure planarity. A condition that fully characterizes planarity was given in 1930's by Kuratowski and Wagner in the form of absence of two specific subgraphs, $\mathit{K}_5$ or $\mathit{K}_{3,3}$ \cite{Swami2014,Kuratowski1930,Wagner1937,PlanarityDuality1980}.}

The next important consideration is how to embed a planar graph in $\mathbb R^2$. For this we refer to \cite{nemhauser1999integer}. It turns out that there may exist several ``nonequivalent embeddings'' \cite{whitney1992congruent,isomorphicembeddings2003}. Interestingly, as we will see, network circulation depends on the particular embedding.

Every planar graph can be drawn on a sphere (and vice versa) via sterographic projection. This amounts to identifying points on $\mathbb R^2$ with points on the (Riemann) sphere by corresponding the ``north pole'' with the ``point at $\infty$'' and any other pair in line with the north pole (the line containing a point on the sphere and the corrsponding projection on $\mathbb R^2$).

For any planar graph $\mathcal G$, and any face $\mathit f$ of $\mathcal G$, the graph can be redrawn on the plane in such a way that $\mathit f$ is the ``outside face'' of $\mathcal G$. 
This can be effected by rotating the projection of the graph onto the Riemann sphere so that the image of the face contains the north pole. But, besides sliding and rotating the graph projection on the Riemann sphere, other transformations are possible that may change the local ordering of vertices, leading to non-equivalent embeddings of the graph on $\mathbb R^2$.

More precisely, for our purposes, two graph embeddings are said to be {\em equivalent} if their corresponding projections onto the sphere can be continuously rotated (and the corresponding vertices shifted onto the sphere without crossing edges) so as to match. The equivalence of two graphs is exemplified in Fig. \ref{fig:IsomorphicGraphs}. For the reasons we just explained, that the positioning of the north pole leads to equivalent embeddings, there are $m = |\mathcal F|$ isomorphic embeddings for every planar graph, as stated next.

\begin{figure}[H]
	\centering
	\includegraphics[scale=0.24]{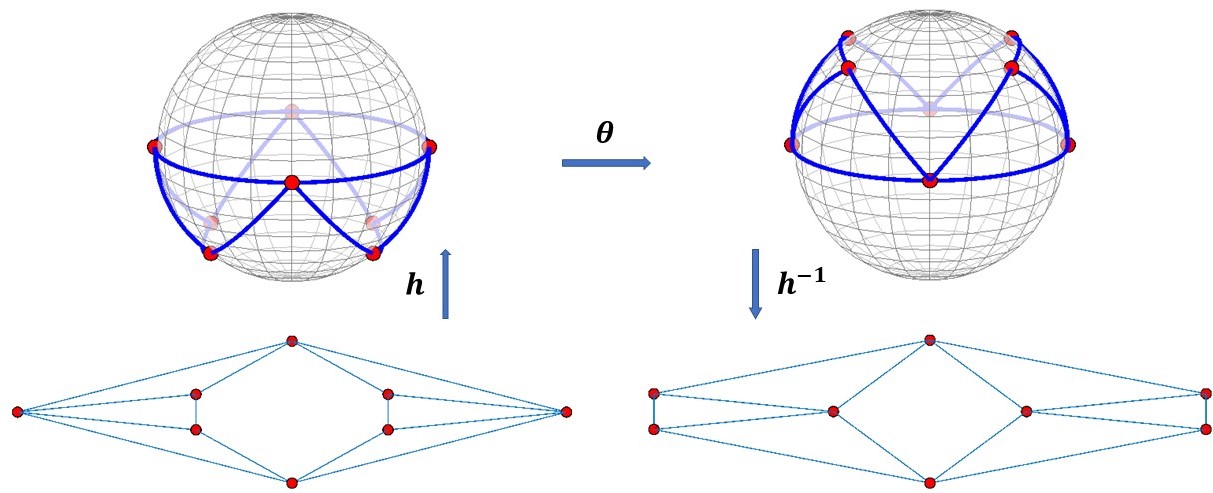}
	\vspace{0.2in}
	\caption{Isomorphic graphs and sequence of graph morphisms; $\theta$ and $\mathit{h}$ are rotation and projection maps, respectively.}
	\label{fig:IsomorphicGraphs}
\end{figure}
\begin{lemma}
There are $m = |\mathcal F|$ isomorphic embeddings for every planar graph.\\[-.35in]
\end{lemma}
\begin{proof}
A sphere is an oriented surface with a given orientation, the ``inside'' and the ``outside.'' When a graph $\mathcal G$ is projected on a sphere, it is embedded on the ``outside.'' Assume a graph $\mathcal G$ with a given embedding, i.e., an orientation system $\mathcal{R}$, $\mathcal{R}:v \rightarrow (e_1, ...,e_k)$, where $v$ is a vertex and $e_i$'s are the incident edges. The projection $\mathit{h} : \mathbb{R}^2 \rightarrow S^2$ is an isomorphism that maps $\mathcal G$ to $\mathcal G_s$ on sphere $S^2$. If $\mathcal G_s$ is then viewed from the inside, the rotation system, $\mathcal{R}$, will be reversed.
Now, assume a homeomorphism $\theta : S^2 \rightarrow S^2$ that rotates $\mathcal G_s$ on the outside of $S^2$. Therefore, $\mathit{h}^{-1} \circ \theta \circ \mathit{h}$ is an isomorphism.
\end{proof}

\subsection{Helmholtz-Hodge decomposition}
\label{sec:potentialCalculation}

We now turn to a brief overview of concepts of vector fields.
The Helmholtz-Hodge decomposition simplifies the analysis by bringing up important properties such as incompressibility and vorticity that can thereby be studied directly
\cite{Helmholtz1858,Hodge1941}.

According to the \textit{Helmholtz-Hodge decomposition Theorem}, the space of a vector field can be uniquely decomposed into mutually $\textit{L}^2$-orthogonal sub-spaces using potential functions \cite{bhatia2012helmholtz}.
These components can be calculated as the gradient of a scalar potential $\phi$ and curl of a vector potential \boldsymbol{$\psi$}, namely,
\[
\mathcal{W}=\mathrm{\nabla }\phi +\mathrm{\nabla }\times \boldsymbol{\psi}+h,
\] 
where $\mathrm{\nabla }\phi $ is the curl-free component (i.e., $\mathrm{\nabla }\times\mathrm{\nabla }\phi=0$) of the vector field $\mathcal{W}$, and $\mathrm{\nabla }\times \boldsymbol{\psi}$ is its divergence-free component (i.e., $\mathrm{\nabla}\cdot \mathrm{\nabla }\times \boldsymbol{\psi}=0$), whereas the harmonic component $h$ is both divergence-free and curl-free.
 
\noindent 
\sbprgraph{\textbf{The curl-free component:}}
\noindent Since the divergence of a curl is zero, we can compute $\phi$, the scalar potential associated with the curl-free component of the vector field $\mathcal{W}$, as the solution of the following Poisson equation.
\[\mathrm{\nabla }\cdot \mathcal{W}=\mathrm{\nabla }\mathrm{\cdot }\mathrm{\nabla }\phi ={\mathrm{\nabla }}^2\phi \]
\noindent
where the last equality holds because the divergence of a gradient is the Laplacian.

\noindent 
\sbprgraph{\textbf{The divergence-free component:}}
\noindent Since the vorticity (normal component of the surface curl) of a gradient field vanishes, the following identity holds: $\hat{n}\ .\ (\mathrm{\nabla }\times \mathcal{W})=\hat{n}\ .\ (\mathrm{\nabla} \times \psi )$. Using the fact that vorticity of the surface curl of a scalar potential is just the surface Laplacian of the potential, we have
\begin{equation} \label{Poisson psi}
\hat{n}\ .\left(\mathrm{\nabla } \times \mathcal{W}\right)=\Delta \psi \
\end{equation}
\noindent
where $\hat{n}$ is the normal vector. Equation (\ref{Poisson psi}) is obeyed if the scalar field $\psi$ is a solution of the above Poisson equation.

In the case of vector fields on $\mathbb R^2$, the curl can be expressed as $\Psi=J\nabla \psi$, where $J$ is an antisymmetric matrix and $\psi$ a scalar potential.
It follows (Stokes' theorem) that the flux crossing any curve connecting two points $a$ and $b$ on $\mathbb R^2$ is given by the difference of the endpoint potentials. As a consequence we have the following.

\begin{thm} \label{thm:maxFluxPotential}
The flux across the path connecting the extrema of a curl potential field is maximum.\\[-.35in]
\end{thm}
\begin{proof}
With $J$ the operator that rotates a vector on $\mathbb R^2$ counter-clockwise by ${\pi }/{2}$, the flux across any path linking a and b is
\begin{equation}
\hspace*{-5pt}I\mathrm{\ =\ }\int^b_a{J\mathrm{\nabla }\psi \ .\ Jds\mathrm{=\ }\int^b_a{\mathrm{\nabla }\psi \ .\ ds\mathrm{=\ }\psi \left(b\right)-\ \psi \left(a\right)}}.
\label{FluxBetweenAnB}
\end{equation}
Hence,
\begin{equation}
\max(I)\mathrm{\ =\ }\max_{a,b \in \mathbb{R}^2}(\psi(b)-\psi(a)).
\label{eq:MaxCirculation}
\end{equation}
\end{proof}

The analogue of Theorem \ref{thm:maxFluxPotential} over discrete vector fields on graphs is discussed in the next section.

\subsection{Planar Net-Flux Graph}
\label{markoviangraph}
Starting from an antisymmetric net flux matrix $F=[F_{ij}]_{i,j}$ in \eqref{netflux} of Markov chain on a planar graph, we consider the graph with adjacency matrix having the zero pattern of $F$; the space of vertices and edges are the collection of the nodes and edges, respectively, that have corresponding non-zero elements in $F$, 
\[\mathcal{V}_F=\left\{v_i\in \mathcal{V} \mid F_{ij}\neq 0\right\},\hspace{0.05in}  \mathcal{E}_F=\left\{e_{ij}\in \mathcal{E} \mid F_{ij}\neq 0\right\}.\]
\noindent
In addition we specify a sign function $\sigma : \mathcal{E}_F \times \mathcal{V}_F \rightarrow \{-1,1\}$ that assigns an orientation, specifically $\sigma = \mathrm{sign}(\mathrm{F_{ij}})$ for all non-zero elements of the net flux matrix, and define the digraph $\mathcal G_{F}(\mathcal{V}_{F},\mathcal{E}_F,\sigma)$. The vector of edge flow weights
\vspace{-0.2cm}
\[
\mathcal{W} = \left(w_{ij}\right)_{i,j}
\]
corresponding to edges $e_{ij}\in \mathcal{E}_F$ with values $w_{ij} = |F_{ij}|$
represents the flow field.
The space of all flow fields is denoted by $\mathcal{U}_{F}$
and assumes a Helmholtz-Hodge decomposition,
\vspace{-0.2cm}
\[\mathcal{U}_F = \mathcal{U}_F^{\rm curl}\oplus {\mathcal{U}_F^{\rm harmonic} \ }\oplus \mathcal{U}_F^{\rm gradient},
\]
where $\mathcal{U}_F^{\rm gradient}$ and $\mathcal{U}_F^{\rm curl}$ are curl-free and divergent-free components. If $\mathcal{W}_{\rm curl},\mathcal{W}_{\rm harmonic},\mathcal{W}_{\rm gradient}$ denote projections of $\mathcal{W}$ in the respective components, then clearly $\mathcal{W}_{\rm gradient}=\mathbf{0}$, since  by assumption $\mathrm{F}$ has no ``sources.''

We wish to capture circulation in a similar manner as in planar flow fields and thereby we seek a curl potential 
$\psi$. The harmonic component $\mathcal{W}_{\rm harmonic}$ relates to circulation about ``holes'' (non-triangular faces \cite{Lim2015}) in the graph.
Thus, before we proceed, we triangulate $\mathcal G_{F}$, and generate a new graph $\mathcal G_{F}^{\rm chordal}$ so as to remove holes and ensure that the harmonic component is zero.

\subsection{Triangular Planar Graph}
\label{subsec:GraphTriangulation}
The curl component of the flow field is defined on triangles, i.e., cycles of length 3 \cite{Lim2015}. The cycle graphs of more than 3 vertices are considered as holes and their edge flow as harmonic components of the flow field.  

To remove the harmonics, we replace the holes that are not bounded by triangles with chordal
subgraphs. That is, we add minimum number of chords, which are the edges with zero flux that are not part of the cycle but each connects two vertices of the cycle. This way, we generate a planar chordal graph such that every chordless cycle subgraph is a triangle. Fig. \ref{fig:triangulateggraphs} shows a triangulated graph with two different embeddings (original graphs are shown in Fig. \ref{fig:embedding}). It can be seen that graph triangulation strongly depends on the embedding. 

The corresponding potential function $\psi$ is now defined on the graph's faces, and hence, can be assigned to the nodes of the dual graph, $(\mathcal G^{\rm chordal}_{F})^*$. The maximum flux, in complete analogy with \eqref{eq:MaxCirculation}, is then obtained by identifying those vertices of the dual graph with minimum and maximum curl potentials.

\subsection{Non-planar Graphs}
Graphs that cannot be drawn on a plane or sphere without edge crossings, i.e., non-planar, can always be drawn on a surface of higher genus \cite{wilson1979introduction}.
A surface is said to be of genus $g\in\{0,1,2\ldots\}$ if it is topologically homeomorphic to a sphere with $g$ handles \cite{hatcher2005algebraic}. For instance, the genus of a sphere is $0$, and that of a torus is $1$.

Accordingly, we define a graph to be of genus $g$ if it can be drawn without crossings on a surface of genus $g$, but not on one of genus $g-1$. It can be easily seen that $\mathit K_{5}$ and $\mathit K_{3,3}$ are graphs of genus $1$ (toroidal graphs); Fig. \ref{fig:toroidalGraph} exemplifies $\mathit  K_{3,3}$ drawn on the torus $T^2$. The graph is drawn by assigning points and continuous, non-intersecting (except at end-points) paths corresponding to the vertices and edges of the graph, respectively. The dual of a non-planar graph can also be drawn on the torus in a similar fashion to a planar graph.

\begin{figure}[H]
	\centering
	\includegraphics[scale=0.25]{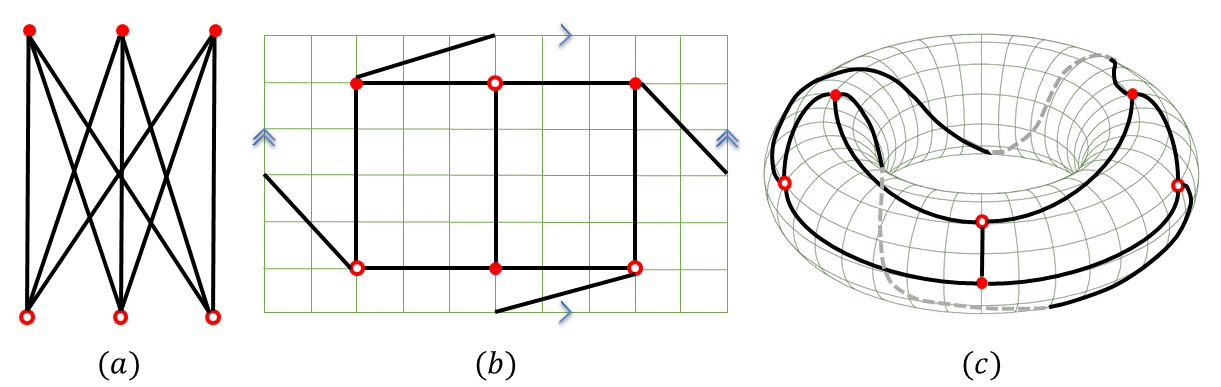}
	\vspace{0.2in}
	\caption{A toroidal graph ($\mathit{K}_{3,3}$): (a) Embedding on $\mathbb{R}^2$; (b) The projective plane; (c) Embedding on a torus.}
	\label{fig:toroidalGraph}
\end{figure}

The decomposition of a continuous vector field on the torus (as a manifold) is more complicated than on a sphere. In this case, the divergence-free component $\mathcal{U}_F^{\rm curl}$, can be partitioned into a toroidal and a poloidal part \cite{ToroidalPoloidal2019,ToroidalPoloidalBerger2018} --a restricted form of the usual Helmholtz decomposition.
We contend that an analogous decomposition of a flow field of genus $1$ graph can be similarly obtained; the corresponding poloidal component is once again generated by a scalar potential (as in the case of planar graphs) whereas the poloidal component is harmonic and represents flux/circulation around holes (of the torus). However, work on graph circulation on genus $1$ graphs is still in progress.

\section{Graph Partitioning}
\label{sec:partitioning}

We summarize the insights gained and highlight the steps needed in Algorithm \ref{algorithm} which helps obtain a 3-partition corresponding to maximum circulation by providing the curl potential $\psi$ on the vertices of its dual graph (i.e., faces of the original graph) and how to numerically calculate this.
The outcome depends on the embedding of $\mathcal G_{F}$, further discussed in Section \ref{sec:embedding}.

\begin{algorithm}
\textbf{Input:}\\
A strongly connected, aperiodic, planar, digraph $\mathcal G$ with a transition matrix $\Pi$.\\
\textbf{Offline Preprocessing:}\\
 1. Calculate the net flux matrix $F$ from \eqref{netflux}.\\
 2. Construct and triangulate $\mathcal G_{F}$ as described in Subsection \ref{markoviangraph}, to generate $\mathcal G^{\rm chordal}_{F}$.\\
 3. Find dual graph $(\mathcal G^{\rm chordal}_{F})^*$.\\
 4. Set the potential $\psi$ for the outside face to zero.\\
\textbf{Computations:}\\
Repeat ($m-1$) times:\\
 1. Find $\psi$ for vertices of the dual graph using \eqref{FluxBetweenAnB} (consider e.g., counter-clockwise as positive direction).\\
\textbf{Output:}\\
Two faces of the primal with potential extrema.
\caption{{\bf Finding curl potential extrema} }
\label{algorithm}
\end{algorithm}

Knowing $\psi$ allows carving 3-partitions that entail maximal circulation. Indeed, any set of two paths on the dual graph between the points of $\psi$-extrema separates the graph in the three regions, $A,B$ and $C$, discussed earlier.
This is summarized next.

\begin{thm}\label{maintheorem}
Consider a divergence-free flow field $\mathcal W$ on the  edges of a strongly connected digraph. Algorithm \ref{algorithm} generates the chordal digraph $\mathcal G_F^{\rm chordal}$ and its dual with an associated curl potential $\psi$. There exist paths in the dual graph connecting two chosen extrema points of $\psi$ that provide a 3-partition with maximal macroscopic circulation.\\
[-.3in]
\end{thm}
\begin{proof}
	Completion of the graph into a chordal graph is the first step of the algorithm and was explained before. We compute $\psi$ as follows.  We assign $0$ at the vertex of the dual of the chordal graph corresponding to the outside face, and proceed to assign values to the remaining vertices of the dual graph so that the difference between values of adjacent vertices equals the (signed, e.g., in the counter-clockwise sense) flux on the corresponding edge of the primal graph. We now explain the last part of the theorem.
	
	Since $\mathcal G^{\rm chordal}_{F}$ is triangular, its dual is 3-edge-connected \cite{gross2003handbook}. By Menger theorem \cite{menger1927allgemeinen, doczkal2019short}, for 3-edge-connected graphs every pair of vertices has 3 edge-disjoint paths in between. Now consider a pair of  vertices on the dual graph, $ v_{min}^{*}$ and $ v_{max}^{*} $,  corresponding to the minimum and maximum of $\psi$, respectively. There are 3 edge-disjoint paths $\mathit{P}_1$, $\mathit{P}_2$, and $\mathit{P}_3$, connecting $ v_{min}^{*}$ and $v_{max}^{*} $. These paths can generate three cycles as follows.
	\[
	\mathit{C}_{ij}=\mathit{P}_{i} \cup \mathit{P}_{j}, \hspace{0.1cm} (1 \leq i < j \leq 3)
	\]
	\noindent
	where $\mathit{C}_{ij} \in \mathcal C$, and $\mathcal C$ is a family of cycles in $(\mathcal G^{\rm chordal}_{F})^*$.
	
	If paths $\mathit{P}_1$, $\mathit{P}_2$, and $\mathit{P}_3$ are also internally disjoint, then they will generate three cycles $\mathit{C}_{12}$, $\mathit{C}_{23}$, and $\mathit{C}_{13}$. If $\mathit{C}_{12}$ and $\mathit{C}_{23}$ are contractible, by 3-path condition \cite{thomassen1990embeddings,mohar2001graphs,de2013algorithms}, $\mathit{C}_{13}$ is also contractible. We only need to show that if $\mathit{P}_2$ lies between $\mathit{P}_1$ and $\mathit{P}_3$, then $\mathit{C}_{13}$ is surface separating and $int(\mathit{C}_{13}) \footnote{ $int(\cdot)$ and $ext(\cdot)$ denote interior and exterior, respectively.}=int(\mathit{C}_{23}) \cup int(\mathit{C}_{12}) \cup \mathit{P_2}$. It follows from Euler's formula that the genus of $int(\mathit{C}_{13}) \cup \mathit{C}_{13}$ is zero, and thus $\mathit{C}_{13}$ is contractible.
	
	Since $\mathit{C}_{ij}$'s are also \textit{discrete Jordan curves} and equivalent to the bonds of $\mathcal G^{\rm chordal}_{F}$ \cite{wilson1979introduction,Chen2013ANO}, by properly selecting interior or exterior of those cycles, these delineate three disjoint and connected sets of primal vertices. Then, by \eqref{eq:MaxCirculation}, the flow that crosses their shared boundaries (i.e., $\mathit{P}_1$, $\mathit{P}_2$, and $\mathit{P}_3$) is maximal.
	
	If $\mathit{P}_1, \mathit{P}_2, \mathit{P}_3$ are edge-disjoint paths between $v_{min}^{*}$ and $v_{max}^{*}$, but not internally disjoint, there exists cycles $\mathit{C}_{ij}$'s that are self-intersecting. That is, the $int(\mathit{C}_{ij})$ is not path-connected and, hence, not contractible.
\end{proof}

\begin{figure}[H]
	\centering
	\includegraphics[scale=0.29]{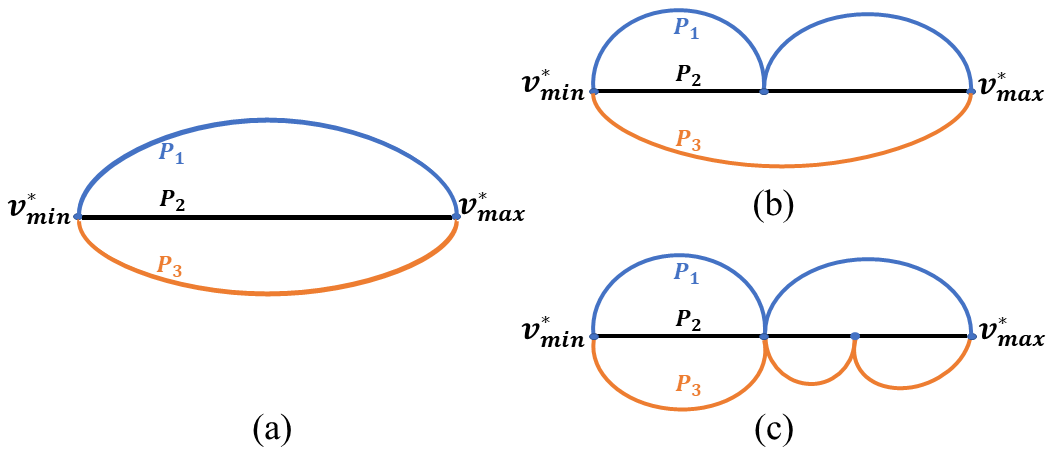}
	\vspace{0.1in}
	\caption{Cases of 3 edge-disjoint paths between $v_{min}^{*}$ and $ v_{max}^{*}$: $\mathit{P}_1, \mathit{P}_2, \mathit{P}_3$ are (a)  internally disjoint, (b,c)  not internally disjoint.}
	\label{fig:3edgedisjointpaths}
\end{figure}
Fig.~\ref{fig:3edgedisjointpaths} exemplifies cases discussed in the above proof.
In case (a) the three cycles \(\mathit{C}_{12}=\mathit{P}_1 \cup \mathit{P}_2, \mathit{C}_{23}=\mathit{P}_2 \cup \mathit{P}_3\), and  \(\mathit{C}_{13}=\mathit{P}_1 \cup \mathit{P}_3\) are contractible. Therefore, all 3-partitions, $int(\mathit{C}_{12}), int(\mathit{C}_{23}), ext(\mathit{C}_{13})$ are connected. The cases (b,c) correspond to cycles that are self-intersecting. Specifically, in Fig. \ref{fig:3edgedisjointpaths}(b) only $\mathit{C}_{12}$ is self-intersecting, whereas  in Fig. \ref{fig:3edgedisjointpaths}(c) all three cycles are self-intersecting. Accordingly, the corresponding 3-partitions may not be connected.\\[-.05in]

\begin{cor}
If there exist 3 edge-disjoint paths that are vertex-disjoint, then each 3-partition is connected. 
\end{cor}
	
Fig. \ref{fig:dualPartitioning} exemplifies Theorem \ref{maintheorem} for a planar graph with two non-equivalent embeddings, where $\mathit{P}_1, \mathit{P}_2, \mathit{P}_3$ are marked with  different colors and the 3-partitions are marked using different colors (red, blue, green) and shapes ($\square$, {\Large $\circ$}, $\triangle$). In the next section we explain how the output of algorithm \ref{algorithm}, and consequently the partitioning, depends on the embedding.

\section{Effect of Embedding on Partitioning}\label{sec:embedding}
Whitney showed that 3-connected graphs have unique embedding, and consequently unique dual graph \cite{whitney1992congruent}. But in general it is possible that if we consider two different embeddings $\mathcal G_{1}, \mathcal G_{2}$ of a planar graph $\mathcal G$, the duals $\mathcal G^{*}_{1}, \mathcal G^{*}_2$ become non-isomorphic, Table \ref{G1G2}. And this may result into a different output of algorithm \ref{algorithm}. 

\begin{table}
\caption{Two embeddings of a graph with non-isomorphic corresponding duals.}
\label{G1G2}
\setlength{\tabcolsep}{5mm} 
\def\arraystretch{1.25} 
\centering
\scalebox{0.7}{%
\begin{tabular}{cccc}
\hline
\\[-0.2cm]
  
\begin{tikzpicture}
\tikzstyle{every node}=[draw, circle, fill=black, inner sep=0pt,minimum size=3pt] 
\draw (3,1) node(1){} -- (3,-1) node(2){} -- (1,-1) node(3){} -- (0,0) node(4){}  -- (1,1) node(5){}-- cycle;
\draw (5) -- (3);
\draw (5) -- (2,0) node(6){};
\draw (1) -- (6);
\draw (2) -- (6);
\draw (3) -- (6);
\end{tikzpicture} & & &\begin{tikzpicture}
\tikzstyle{every node}=[draw, circle, fill=black, inner sep=0pt,minimum size=3pt] 
\draw (0,0) node(1){} -- (1,-1) node(2){} -- (3,0) node(3){} -- (1,1) node(4){}-- cycle;
\draw (2) -- (4);
\draw (4) -- (2,0.2) node(5){};
\draw (5) -- (2,-0.2) node(6){};
\draw (2) -- (6);
\draw (3) -- (6);
\draw (3) -- (5);
\end{tikzpicture}
 
\\
    $ \mathcal G^{*}_{1}$& & & $\mathcal G^{*}_{2} $ 

 \\ [0.3cm]\hline
 
  \end{tabular}}
\end{table}

\subsection{Example}
\label{sec:examples}
Given transition matrix for a planar graph as in \eqref{Pi}, netflux matrix is calculated using equation \eqref{netflux}. Fig. \ref{fig:embedding} indicates two possible graphs which are constructed based on calculated net flux matrix in \eqref{netfluxexample}. These two embeddings of a connected planar graph are related by flipping at separating pair\footnote{In graph theory, a vertex separator for nonadjacent vertices a and b is a vertex subset  $S\subset V$ such that the removal of S from the graph separates a and b into distinct connected components.}. As described in \ref{markoviangraph}, graphs are triangulated, $\mathcal G'_{1}, \mathcal G'_{2}$ in Fig. \ref{fig:dualPartitioning}.\\
\begin{equation}\label{Pi}
 \resizebox{0.8\hsize}{!}{%
$\Pi = \left[\begin{matrix} 0 & 0.25 & 0 & 0 & 0.25 & 0 & 0.25 & 0.25\\ 0.333 & 0 & 0.333 & 0.333 & 0 & 0 & 0 & 0\\0 & 0.25 & 0 & 0.25 & 0 & 0.25 & 0 & 0.25\\0 & 0.333 & 0.333 & 0 & 0.333 & 0 & 0 & 0\\0.333 & 0.333 & 0 & 0.333 & 0 & 0 & 0 & 0\\0 & 0 & 0.333 & 0 & 0 & 0 & 0.333 & 0.333\\0.5 & 0 & 0 & 0 & 0 & 0.5 & 0 & 0\\0.25 & 0 & 0.25 & 0 & 0 & 0.25 & 0.25 & 0 \end{matrix} \right]$
}
\end{equation}
\begin{equation}\label{netfluxexample}
 \resizebox{1\hsize}{!}{%
$F = \left[\begin{matrix} 0 & -0.0076 & 0 & 0 & 0.014 & 0 & -0.016 & 0.0096\\ 0.0076 & 0 & 0.0082& 0.0099 & -0.0256 & 0 & 0 & 0\\0 & -0.0082 & 0 & 0.0017 & 0 & -0.0025 & 0 & 0.009\\0 & -0.0099 & -0.0017 & 0 & 0.0116 & 0 & 0 & 0\\-0.014 & 0.0256 & 0 & -0.0116 & 0 & 0 & 0 & 0\\0 & 0 & 0.0025 & 0 & 0 & 0 & -0.014 & 0.0115\\0.016 & 0 & 0 & 0 & 0 & 0.014 & 0 & -0.03\\-0.0096 & 0 & -0.009 & 0 & 0 & -0.0115 & 0.03 & 0 \end{matrix} \right]$%
}
\end{equation}

\begin{table*}[t]
\captionof{figure}{Two possible embeddings, constructed based on netflux matrix in \eqref{netfluxexample}.}
\label{fig:embedding}
\setlength{\tabcolsep}{5mm} 
\def\arraystretch{1.25} 
\centering
\scalebox{0.5}{%
\begin{tabular}{cccc}
\\[-0.2cm]
\begin{tikzpicture}
\tikzset{vertex/.style = {shape=circle,draw,minimum size=0.5em}}
\tikzset{edge/.style = {->,> = latex'}}
\node[vertex] (2) at  (0,0) {2};
\node[vertex] (4) at  (2,-0.6) {4};
\node[vertex] (5) at  (2,0.6) {5};
\node[vertex] (1) at  (3,3) {1};
\node[vertex] (3) at (3,-3) {3};
\node[vertex] (6) at (4,-0.6) {6};
\node[vertex] (7) at (4,0.6) {7};
\node[vertex] (8) at (6,0) {8};

\draw[edge] (2) to (4);
\draw[edge] (5) to (2);
\draw[edge] (4) to (5);
\draw[edge] (2) to (3);
\draw[edge] (3) to (4);
\draw[edge] (1) to (5);
\draw[edge] (2) to (1);
\draw[edge] (7) to (1);
\draw[edge] (7) to (6);
\draw[edge] (6) to (3);
\draw[edge] (1) to (8);
\draw[edge] (8) to (7);
\draw[edge] (6) to (8);
\draw[edge] (3) to (8);
\end{tikzpicture} 
 &&&
\begin{tikzpicture}
\tikzset{vertex/.style = {shape=circle,draw,minimum size=0.5em}}
\tikzset{edge/.style = {->,> = latex'}}
\node[vertex] (2) at  (0,0) {2};
\node[vertex] (4) at  (2,-0.6) {4};
\node[vertex] (5) at  (2,0.6) {5};
\node[vertex] (1) at  (3,3) {1};
\node[vertex] (3) at (3,-3) {3};
\node[vertex] (6) at (6,-0.6) {6};
\node[vertex] (7) at (6,0.6) {7};
\node[vertex] (8) at (4,0) {8};

\draw[edge] (2) to (4);
\draw[edge] (5) to (2);
\draw[edge] (4) to (5);
\draw[edge] (2) to (3);
\draw[edge] (3) to (4);
\draw[edge] (1) to (5);
\draw[edge] (2) to (1);
\draw[edge] (7) to (1);
\draw[edge] (7) to (6);
\draw[edge] (6) to (3);
\draw[edge] (1) to (8);
\draw[edge] (8) to (7);
\draw[edge] (6) to (8);
\draw[edge] (3) to (8);
\end{tikzpicture}
\\
     $\mathcal G_{1}$& & & $\mathcal G_{2}$  
  \\ [0.3cm]
  \end{tabular}}
\end{table*}
\begin{table*}[t]
\vspace*{0.5 cm}
\captionof{figure}{Triangulated of the graphs in Fig. \ref{fig:embedding}.}
\label{fig:triangulateggraphs}
\setlength{\tabcolsep}{5mm} 
\def\arraystretch{1.25} 
\centering

\scalebox{0.5}{%
\begin{tabular}{cccc}
\\[-0.2cm]

\begin{tikzpicture}[
V/.style = {
            draw,circle,thick,fill=#1},
V/.default = black!0
]
                        
\node[V] (2) at  (0,0) {2};
\node[V] (4) at  (2,-0.6) {4};
\node[V] (5) at  (2,0.6) {5};
\node[V] (1) at  (3,3) {1};
\node[V] (3) at (3,-3) {3};
\node[V] (6) at (4,-0.6) {6};
\node[V] (7) at (4,0.6) {7};
\node[V] (8) at (6,0) {8};

 \begin{scope}[V/.default = green!20]
 \node[text width=0.05cm] at (0.8,2) {};
\node[text width=0.05cm] at (1.9,1.4) {};
\node[text width=0.05cm] at (1.3,0) {};
\node[text width=0.05cm] at (1.8,-1.2) {};
\node[text width=0.05cm] at (3,1.4) {};
\node[text width=0.05cm] at (3.5,0.2) {};
\node[text width=0.05cm] at (2.5,-0.2) {};
\node[text width=0.05cm] at (3,-1.4) {};
\node[text width=0.05cm] at (4.2,1.3) {};
\node[text width=0.05cm] at (4.5,0) {};
\node[text width=0.05cm] at (4.2,-1.3) {};

    \end{scope}
\draw (2) to (4);
\draw (5) to (2);
\draw (4) to (5);
\draw (2) to (3);
\draw (3) to (4);
\draw (1) to (5);
\draw (2) to (1);
\draw (7) to (1);
\draw (7) to (6);
\draw (6) to (3);
\draw (1) to (8);
\draw (8) to (7);
\draw (6) to (8);
\draw (3) to (8);
\draw (5) to (7);
\draw (5) to (6);
\draw (4) to (6);   
\end{tikzpicture} 
&&&
\begin{tikzpicture}[
V/.style = {
            draw,circle,thick,fill=#1},
V/.default = black!0
                        ]
\node[V] (2) at  (0,0) {2};
\node[V] (4) at  (2,-0.6) {4};
\node[V] (5) at  (2,0.6) {5};
\node[V] (1) at  (3,3) {1};
\node[V] (3) at (3,-3) {3};
\node[V] (6) at (6,-0.6) {6};
\node[V] (7) at (6,0.6) {7};
\node[V] (8) at (4,0) {8};
\begin{scope}[V/.default = green!20]
\node[text width=0.05cm] at (1.9,1.4) {};
\node[text width=0.05cm] at (1.3,0) {};
\node[text width=0.05cm] at (1.8,-1.2) {};
\node[text width=0.05cm] at (2.9,1.2) {};
\node[text width=0.05cm] at (2.75,0) {};
\node[text width=0.05cm] at (2.9,-1.2) {};
\node[text width=0.05cm] at (4.1,1.4) {};
\node[text width=0.05cm] at (5.3,0) {};
\node[text width=0.05cm] at (4.1,-1.2) {};
\node[text width=0.05cm] at (0.8,2) {};

    \end{scope}
\draw(2) to (4);
\draw(5) to (2);
\draw(4) to (5);
\draw(2) to (3);
\draw(3) to (4);
\draw(1) to (5);
\draw(2) to (1);
\draw(7) to (1);
\draw(7) to (6);
\draw(6) to (3);
\draw(1) to (8);
\draw(8) to (7);
\draw(6) to (8);
\draw(3) to (8);
\draw(5) to (8);
\draw(4) to (8);
              \end{tikzpicture}
             
\\
     $\mathcal G'_{1}$& & & $\mathcal G'_{2}$ 
  \\ [0.3cm]
  \end{tabular}}
\end{table*}
\begin{table*}[t]
\vspace*{0.5 cm}

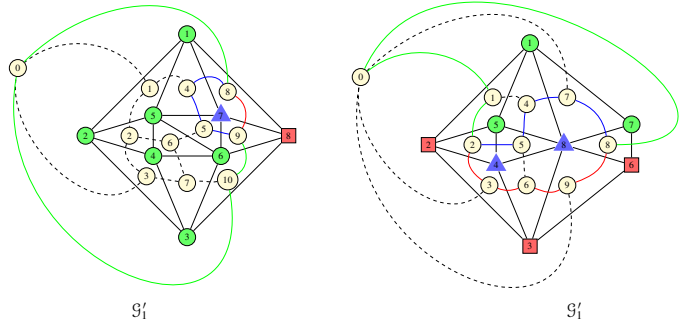
\captionof{figure}{Triangulated and dual of the graphs in Fig. \ref{fig:embedding} and their 3-partitions.}
\label{fig:dualPartitioning}
\setlength{\tabcolsep}{4.5mm} 
\def\arraystretch{1.2} 
\centering
\scalebox{0.45}{%
\begin{tabular}{cc}
\\[-3cm]  
\begin{tikzpicture}[
V/.style = {
            draw,circle,thick,fill=#1, minimum size=1.5em},
V/.default = black!20
                        ]
 \tikzset{vertex1/.style = {shape=rectangle,draw,minimum size=1.5em}}
\tikzset{triangle/.style = {fill=blue!20, regular polygon, regular polygon sides=3, inner sep=0.1em}}

\node[V, fill=green!60] (2) at  (0,0) {2};
\node[V, fill=green!60] (4) at  (2,-0.6) {4};
\node[V, fill=green!60] (5) at  (2,0.6) {5};
\node[V, fill=green!60] (1) at  (3,3) {1};
\node[V, fill=green!60] (3) at (3,-3) {3};
\node[V, fill=green!60] (6) at (4,-0.6) {6};
\node[triangle, fill=blue!60] (7) at (4,0.6) {7};
\node[vertex1, fill=red!60] (8) at (6,0) {8};

 \begin{scope}[V/.default = yellow!20]
\node[V] (125) at (1.9,1.4) {1};
\node[V] (245) at (1.3,0) {2};
\node[V] (out) at (-2,2) {0};
\node[V] (234) at (1.8,-1.2) {3};
\node[V] (368) at (4.2,-1.3) {10};
\node[V] (678) at (4.5,0) {9};
\node[V] (178) at (4.2,1.3) {8};
\node[V] (157) at (3,1.4) {4};
\node[V] (346) at (3,-1.4) {7};
\node[V] (456) at (2.5,-0.2) {6};
\node[V] (567) at (3.5,0.2) {5};
    \end{scope}
\draw (2) to (4);
\draw (5) to (2);
\draw (4) to (5);
\draw (2) to (3);
\draw (3) to (4);
\draw (1) to (5);
\draw (2) to (1);
\draw (7) to (1);
\draw (7) to (6);
\draw (6) to (3);
\draw (1) to (8);
\draw (8) to (7);
\draw (6) to (8);
\draw (3) to (8);
\draw (5) to (7);
\draw (5) to (6);
\draw (4) to (6);
\draw[dashed] 
       (125) to [out=-140,in=90](245) 
       (234) to [out=-135,in=-100](out)
       (245) to [out=-110,in=145](234)
       (125) to [out=120,in=40](out)
       (125) to [out=40,in=120](157)
       (245) to [out=-10,in=170](456)
       (456) to [out=20,in=205](567)
       (456) to [out=-60,in=110](346)
       (368) to [out=190,in=10](346)
       (346) to [out=170,in=-10](234);
  \draw
  [red] (678) to [out=60,in=-30](178);
  \draw 
[blue] (678) to [out=170,in=-10](567)
 (157) to [out=-60,in=110](567)
 (157) to [out=40,in=120](178);
\draw 
[green] (178) to [out=90,in=45, looseness=1.3](out)
  (368) to [out=-80,in=-100,looseness=2](out)
  (678) to [out=-60,in=30](368);

\end{tikzpicture} 
&
\begin{tikzpicture}[
V/.style = {
            draw,circle,thick,fill=#1, minimum size=1.5em},
V/.default = black!20
                        ]
\tikzset{vertex1/.style = {shape=rectangle,draw,minimum size=1.5em}}
\tikzset{triangle/.style = {fill=blue!20, regular polygon, regular polygon sides=3,inner sep=0.1em}}    
\node[V, fill=green!60, minimum size=1.5em] (1) at  (3,3) {1};                    
\node[vertex1, fill=red!60] (2) at  (0,0) {2};
\node[vertex1, fill=red!60] (3) at (3,-3) {3};
\node[triangle, fill=blue!60] (4) at  (2,-0.6) {4};
\node[V, fill=green!60,  minimum size=1.5em] (5) at  (2,0.6) {5};
\node[vertex1, fill=red!60] (6) at (6,-0.6) {6};
\node[V, fill=green!60,  minimum size=1.5em] (7) at (6,0.6) {7};
\node[triangle, fill=blue!60] (8) at (4,0) {8};
\begin{scope}[V/.default = yellow!20]
\node[V] (125) at (1.9,1.4) {1};
\node[V] (245) at (1.3,0) {2};
\node[V] (out) at (-2,2) {0};
\node[V] (234) at (1.8,-1.2) {3};
\node[V] (458) at (2.75,0) {5};
\node[V] (368) at (4.1,-1.2) {9};
\node[V] (678) at (5.3,0) {8};
\node[V] (178) at (4.1,1.4) {7};
\node[V] (158) at (2.9,1.2) {4};
\node[V] (348) at (2.9,-1.2) {6};

    \end{scope}
\draw(2) to (4);
\draw(5) to (2);
\draw(4) to (5);
\draw(2) to (3);
\draw(3) to (4);
\draw(1) to (5);
\draw(2) to (1);
\draw(7) to (1);
\draw(7) to (6);
\draw(6) to (3);
\draw(1) to (8);
\draw(8) to (7);
\draw(6) to (8);
\draw(3) to (8);
\draw(5) to (8);
\draw(4) to (8);
\draw[dashed] 
       (234) to [out=-135,in=-100](out)
       (368) to [out=-80,in=-100,looseness=2](out)
       (125) to [out=0,in=140](158)
       (458) to [out=-80,in=100](348)
       (178) to [out=90,in=45, looseness=1.3](out);
        \draw
        [red]    (245) to [out=-110,in=145](234)
        (348) to [out=-20,in=210](368)
        (678) to [out=-100,in=20](368)
        (234) to [out=-20,in=200](348);
        \draw
        [blue] (245) to [out=0,in=-180](458)
        (158) to [out=-100,in=80](458)
        (158) to [out=20,in=180](178)
        (678) to [out=100,in=-20](178);
        \draw
        [green] 
        (125) to [out=-140,in=90](245)
        (125) to [out=120,in=40](out) 
        (678) to [out=0,in=65,looseness=2](out);
              \end{tikzpicture}
  \\[-2.5cm]
     {\Large $\mathcal G'_{1}$}& {\Large $\mathcal G'_{1}$} 
  \\ [0.3cm]

  \end{tabular}}
\end{table*}

Based on Algorithm \ref{algorithm}, a vector of curl potentials for each triangulated graph is calculated,
\begin{equation}\label{phis}
\begin{aligned}
\resizebox{0.95\hsize}{!}{%
$\psi_{1} $=$ \left[0, -0.0076, 0.0181, 0.0082, 0.0064, 0.0064, 0.0064, 0.0064,-0.0096, 0.0205, 0.009\right]^{T},$%
}
\\
\resizebox{0.95\hsize}{!}{%
$\psi_{2} $=$ \left[0, -0.0076, 0.0181, 0.0082, 0.0064, 0.0064, 0.0064, 0.016, -0.014, -0.0025\right]^{T}.$%
}
\end{aligned}
\end{equation}
where $\psi_{i_{j}}$ corresponds to the curl potential of the $j^{th}$ face of graph $\mathcal G'_{i}$, $i \in \{1, 2\}$.
The faces with maximum and minimum potential for $\mathcal G'_{1}$ and $\mathcal G'_{2}$ are $\{\mathit{f}_{8}, \mathit{f}_{9}\}$ and $\{\mathit{f}_2, \mathit{f}_{8}\}$, respectively. Applying Theorem \ref{maintheorem}, the  3-partitions for them are $\{\{6\}, \{7\}, \{1, 2, 3, 4, 5, 8\}\}$ and $\{\{4, 8\}, \{2, 3, 6\}, \{1, 5, 7\}\}$,
respectively. 

This example highlights that the output of Algorithm \ref{algorithm} and consequently the partitioning of a graph varies with the embedding. Hence, in order to determine faces with extremum potentials we need to specify the embedding along with the transition matrix. To specify the embedding, we need to conduct a rotation system; there exists a unique such rotation system for every locally oriented graph embedding \cite{gross2001topological}.\\

\spacingset{.98}

\bibliographystyle{IEEEtran}
\bibliography{references}
\end{document}